\documentclass[12pt]{amsart}
\usepackage{amsfonts}
\usepackage{graphicx}
\usepackage{tabularx}
\usepackage{array}
\usepackage[usenames,dvipsnames]{color}
\usepackage{comment}
\usepackage{amsmath}
\usepackage{amsthm}
\usepackage{amssymb}
\usepackage{fullpage}
\usepackage[dvipsnames]{xcolor}
\usepackage{listings}
\usepackage{algorithm}
\usepackage{algpseudocode}

\usepackage{caption}
\usepackage{subcaption}
\usepackage{url}

\usepackage[
backend=biber,
style=alphabetic,
]{biblatex}
\newtheorem{theorem}{Theorem}[section]
\newtheorem{proposition}[theorem]{Proposition}

\newtheorem{lemma}[theorem]{Lemma}
\newtheorem{corollary}[theorem]{Corollary}
\theoremstyle{definition}
\newtheorem{definition}[theorem]{Definition}

\def\epsilon{\varepsilon}

\title{On The Structure of EFX Orientations on Graphs}
\author{Jinghan A Zeng}
\address{University of Illinois Urbana-Champaign}
\email{jazeng2@illinois.edu}

\author{Ruta Mehta}
\address{University of Illinois Urbana-Champaign}
\email{rutameht@illinois.edu}

\begin{document}

\begin{abstract}
Fair division is the problem of allocating a set of items among agents in a fair manner. One of the most sought-after fairness notions is envy-freeness (EF), requiring that no agent envies another's allocation. When items are indivisible, it ceases to exist, and envy-freeness up to any good (EFX) emerged as one of its strongest relaxations. The existence of EFX allocations is arguably the biggest open question within fair division. Recently, Christodoulou, Fiat, Koutsoupias, and Sgouritsa (EC 2023) showed that EFX allocations exist for the case of graphical valuations where an instance is represented by a graph: nodes are agents, edges are goods, and each agent values only her incident edges. On the other hand, they showed NP-hardness for checking the existence of EFX orientation where every edge is allocated to one of its incident vertices, and asked for a characterization of graphs that exhibit EFX orientation regardless of the assigned valuations. 

In this paper, we make significant progress toward answering their question. We introduce the notion of {\em strongly EFX orientable graphs} -- graphs that have EFX orientations regardless of how much agents value the edges. We show a surprising connection between this property and the chromatic number $\chi(G)$ of the graph $G$. In particular, we show that graphs with $\chi(G)\le 2$ are strongly EFX orientable, and those with $\chi(G)>3$ are not strongly EFX orientable. We provide examples of strongly EFX orientable and non-strongly EFX orientable graphs of $\chi(G)=3$ to prove tightness. Finally, we give a complete characterization of strong EFX orientability when restricted to binary valuations. 

\end{abstract}
\maketitle
\newpage

\section{Introduction}
Fair division is a problem of allocating a set $M$ of goods among $n$ agents in a {\em fair} manner \cite{Steinhaus48}. This is an age-old problem with numerous contemporary applications, e.g., division of family inheritance \cite{PrattZ90}, divorce settlements \cite{BramsT96}, spectrum allocation \cite{EtkinPT05}, air traffic management \cite{Vossen02}, course allocation \cite{BudishC10} and many more \footnote{See \url{www.spliddit.org} and \url{www.fairoutcomes.com} for a detailed discussion on fair division protocols used in day-to-day life.}. In this paper, we study the fair division of indivisible goods. In this case, the preferences of each agent $i \in [n]=\{1,2...n \}$ are represented by a monotone increasing valuation function $f_i:2^{M} \to \mathbb{R}_{\geq 0}$, specifying how much they value each subset of goods \footnote{Using $v_i$ to denote the valuation function is the convention in the literature, however, we will use $f_i$ to avoid confusing notation, as $v$ is also used to denote a vertex in a graph.}. An allocation of goods is a partition of $M$ into $n$ subsets $X=(X_1,X_2...X_n)$, where $X_i$ is allocated to agent $i$.

One of the most sought-after fairness notions is of {\em envy freeness}. An allocation is said to be {\em envy free (EF)} if no agent $i$ envies another's allocation: $f_i(X_i) \ge f_i(X_j), \forall i,j\in[n]$. With indivisible items, an envy-free allocation may not always exist, for example, allocating one iPhone among two agents who both value it highly. Hence, several relaxations of envy-freeness have been studied, arguably the strongest of which is \emph{envy-freeness up to any good (EFX)} \cite{Caragiannis19}. An allocation is said to be EFX if for any two agents $i,j$, $i$ does not envy any proper subset of $j$'s bundle. That is, for any $X \subset X_j$, $f_i(X_i) \geq f_i(X)$. In other words, $i$ does not envy $X_j$ after the removal of any good from $X_j$. 

Determining whether an EFX allocation exists is arguably the biggest open question within fair division \cite{ProcaciaCACM}. 
For the case of additive valuations, the existence of EFX is known for three agents~\cite{Chaudhury2020}, and beyond additive only for the case of two agents~\cite{PR18}. Given the notoriety of this problem, there has been extensive work exploring special cases and relaxations (see Section \ref{sec:RW} for an overview of related works). One such prominent special case is of {\em graphical valuations} considered by Christodoulou, Fiat, Koutsoupias, and Sgouritsa \cite{Christodoulou2023}. Here, an instance is represented by an undirected graph $G=(V,E)$, where the vertices correspond to agents and the edges correspond to goods, such that each edge is valued positively only by its incident agents. In other words, each vertex $v \in V$ has a valuation function $f_v: 2^{E} \to \mathbb{R}_{\geq 0}$, satisfying the property that $f_v(X) = f_v(X \cap E(v))$, where $E(v)$ is the set of edges incident to $v$, for all $X \subseteq E$. 

An {\em orientation} of edges in $G$ can be interpreted as an allocation where each edge/good is given to its incident vertex/agent the edge is oriented towards. It is called an {\em EFX orientation} if the corresponding allocation is EFX. \cite{Christodoulou2023} showed that an EFX orientation may not always exist and, furthermore, proved that determining whether a graph has an EFX orientation for a given set of valuations is NP-hard. 

On the other hand, they showed that if we allow edges to be allocated to non-incident vertices, then an EFX allocation exists. However, they noted that {\em EFX orientations} are more desirable since they avoid lossy assignments (assigning an edge to a non-incident vertex). Motivated by this, they asked the following question \cite{Christodoulou2023}: 


\begin{quote}
    ``[A] question of interest is understanding for what classes of graphs an EFX orientation is guaranteed to exist. E.g., an EFX orientation always exists in trees, cycle graphs, and multistars.'' 
\end{quote}

\subsection{Our Contributions} In this paper, we make significant progress toward answering the above question. To capture their question systematically, we introduce the notion of \emph{strongly EFX orientable} graphs, which are graphs that have an EFX orientation regardless of valuation. All graphs are simple (i.e. no self loops or multiedges) unless stated otherwise. We show that a characterization of these strongly EFX orientable graphs has a surprising connection to the \emph{chromatic number} of the graph, where the chromatic number, denoted $\chi(G)$, is defined as the minimum number of colors needed to color the vertices of the graph such that no two adjacent vertices have the same color. 
In particular, we show that:

\begin{itemize}
    \item Any graph of chromatic number $\chi(G) \leq 2$ is strongly EFX-orientable.
    \item All strongly EFX-orientable graphs have chromatic number $\chi(G) \leq 3$. 
    \item There exist graphs of $\chi(G)=3$ which are not strongly EFX-orientable, as well as graphs with $\chi(G)=3$ which are strongly EFX-orientable, so this bound is sharp. 
    \item For the case of binary valuation function, we give a complete characterization. Given a graph $G$, the following two statements are equivalent:
    \begin{itemize}
        \item For any 0-1 additive valuation assigned to $G$, $G$ has an EFX orientation.
        \item For every subgraph $H \subseteq G$ such that $H$ is a forest consisting of trees $T_1, T_2...T_k$, for every $1 \leq i \leq k$ there exists $x_i \in T_i$ such that $\bigcup_{i=1}^k N_H(x_i)$ forms an independent set on $G$. 
    \end{itemize}


\end{itemize}

\subsection{Further Related Work}\label{sec:RW} Fair division has been extensively studied, with substantial work dedicated to understanding the existence of EFX allocations for, including but not limited to, special cases, approximation, EFX with charity, and efficiency. Below, we give a brief overview. 
\medskip

\noindent{\bf Special Cases.} In addition to the cases of two \cite{PR18} and three agents \cite{Chaudhury2020, Akrami23}, several special cases have been studied: \cite{berger2022almost} showed existence for four agents where one good may remain un-allocated (goes to charity). For the case of binary valuations, {\em i.e.,} every item is valued at $0$ or $1$, \cite{BuSongYu23} showed EFX existence with arbitrary many agents. \cite{LivanosMM22} extended this to restricted-additive valuations where every agent values good $j$ at $0$ or $v_j>0$, {\em i.e.,} $f_{i}(j) \in \{0,v_j\},\ \forall i\in N$. Extending the case of the identical valuations, \cite{Mahara23} showed that EFX exists if all the agents have one of two given valuation functions. \cite{GhosalPNV23} considered the case where all but two agents have identical valuation functions. Additionally, \cite{Mahara21} showed that EFX exists for $n$ agents when there are at most $n+3$ items.  \cite{GorantlaMV23} showed that EFX exists when there are $2$ types of objects and all agents have the same value for objects of the same type. 

\medskip


\noindent{\bf Relaxations of EFX: Approximation and Charity.}
EFX has been studied under several relaxations as well. The most notable of these are EFX with charity, and approximate EFX. \cite{PR18,chan2019maximin} gave an algorithm to compute $0.5$-approximate EFX, which was improved to $0.68$ by \cite{AmanatidisMN20}. \cite{CaragiannisGH19} showed the existence of EFX allocations where some items go to charity (remain unallocated) with $1/2$-approximate Nash Welfare guarantee. \cite{CKMS21} showed the existence of EFX allocations where at most $(n-1)$ items go to charity and no agent envies the charity. This was improved by \cite{berger2022almost} to $(n-2)$ charity. \cite{PR18, AmanatidisMN20} studied algorithms to find approximate EFX allocations. A series of works \cite{chaudhury2021improving,akrami2022efx} combined both the relaxations to get a $(1-\epsilon)$-approximate EFX with $\tilde{O}(\sqrt{|N|})$ charity. 
Another popular relaxation of envy-freeness is \emph{envy-freeness up to one good (EF1)} where no agent envies another agent following the removal of \emph{some} good from the other agent's bundle. The existence of EF1 allocations is well-known for any number of agents, even when agents have general monotone valuation functions~\cite{LiptonMMS04}. 
\medskip


\noindent{\bf Efficiency with Fairness.} Efficiency alongside fairness is another requirement that is extensively studied. Two of the most popular measures of efficiency are Pareto-optimality and Nash welfare. \cite{Caragiannis19} showed that any allocation that has the maximum Nash welfare is guaranteed to be Pareto-optimal (efficient) and EF1 (fair). \cite{BKV18} gave a pseudo-polynomial algorithm to find an allocation that is both EF1 and Pareto-optimal. 
Other works explore relaxations of EFX with high Nash welfare~\cite{CaragiannisGH19, CKMS21,feldman2023}. 
\medskip

\noindent{\bf Organization.} In Section \ref{sec:prel} we formally define the fair division problem and strong EFX-orientability. Section \ref{sec:nec} discusses the necessary conditions for strong EFX-orientability and obtains the upper bound of 3 on the chromatic number. In the process it completely characterizes 0-1 strong  EFX-orientability. Section \ref{sec:suf} discusses the sufficiency conditions through bipartiteness or near-bipartiteness of the graph. Section \ref{sec:fs} discusses structures in certain 3-chromatic graphs that prevent strong EFX-orientability, leaving 3-chromatic graphs as an ambiguous case.  Finally, we conclude with a brief discussion in Section \ref{sec:disc}.


\section{Preliminaries}\label{sec:prel}

A discrete fair division instance is given by $([n], M, \mathcal{F})$, where $[n]$ is the set of agents, $M$ is the set of goods, and $\mathcal{F} = \{ f_1,f_2 ... f_n \}$ is the set of valuation functions, one for each agent. For each $i\in [n]$, $f_i:2^M \rightarrow \mathbb{R}_{\geq 0}$ represents agent $i$'s value over the bundles (sets) of goods. Thus, $f_i$ is non-negative and \emph{monotone}, {\em i.e.,} for $A,B \subseteq M$, $A \subseteq B$ implies $0\le f_i(A) \leq f_i(B)$. A valuation function $f_i$ is said to be \emph{additive} if for every subset $X=\{m_1, m_2, ...m_k\} \subseteq M$, $f_i(X) = f_i(\{ m_1 \}) + f_i(\{ m_2 \}) + ... +f_i(\{ m_k \})$, and a valuation is \emph{0-1 additive} if it is additive and $f_i(\{m \}) \in \{ 0,1 \}$ for all $m \in M$. 

An \emph{allocation} $X=(X_1,X_2...X_n)$ is a partition of $M$ where agent $i$ is assigned the bundle $X_i$. An allocation is said to be {\em envy-free (EF)}, if $f_i(X_i) \ge f_i(X_j)$ for all $i,j\in [n]$. Since such an allocation may not exist when goods are indivisible, we consider the following relaxation:

\begin{definition}
An allocation is {\em envy-free up to any good (EFX)} if for all $i,j\in[n]$, $f_i(X_i) \geq f_i(X_j - \{g\})$ for all $g \in X_j$.     
\end{definition}

{\bf Graphical Valuations.} The graphical version of a discrete fair division setting is represented by $(G, \mathcal{F})$, where $G=(V,E)$ is an undirected graph with vertices being the agents $[n]$ ($V=[n]$), and edges being the goods ($E=M$). Every agent vertex values only her incident edges, {\em i.e.,} her valuation functions $f_i=2^{E} \rightarrow \mathbb{R}_{\geq 0}$ has the property that $f_i(X) = f_i(X \cap E(i))$ for all $X \subseteq E$, where $E(i)$ is the set of edges incident to $i$. 

An \emph{orientation} of a graph $G$ assigns each edge $e \in G$ an incident vertex as the head and the other incident vertex as a tail. We say that an orientation gives an allocation $X_1,X_2...X_n$, where $X_i$ is the set of all edges that have vertex $i$ as their head. In this paper, we will be mainly dealing with two types of strong EFX-orientability, the general case and the binary case.  

\begin{definition}
    A graph is \emph{strongly EFX-orientable} if, for any assigned monotone valuation, there exists an EFX orientation. 
\end{definition}

\begin{definition}
    A graph $G$ is \emph{0-1 strongly EFX-orientable} if for any additive 0-1 valuation on the edges, $G$ has an EFX-orientation. 
\end{definition}

\section{Strong EFX Orientability: Necessary Condition via Tripartiteness}\label{sec:nec}
In this section, we prove the chromatic number upper bound of 3 on strongly EFX-orientable graphs. We will do this by analyzing graphs of chromatic number greater than 3, and finding a bad valuation on the edges that makes an EFX orientation impossible. It turns out that it is enough to only look at situations where the valuations on the edges are either 0 or 1. Therefore, we will start by obtaining a complete characterization of \emph{0-1 strongly EFX-orientable} graphs. 
\medskip
\medskip

\subsection{Characterization of 0-1 Strongly EFX-Orientable Graphs}
We first provide a complete characterization of 0-1 strongly EFX-Orientable graphs, and then provide a simpler necessary condition. These together will let us prove the upper bound of $3$ on the chromatic number of 0-1 strongly EFX-orientable graphs, and thereby on the strongly EFX-orientable graphs as well. 

The condition that completely characterizes 0-1 strongly EFX-orientable graph depends on certain tree structures and independent set conditions described in the following lemma. The condition is a bit of a mouthful, so we demonstrate it for $K_{2,4}$ graph in Figure \ref{fig:2}.

\begin{lemma}\label{lem:01char}
    A graph $G$ is 0-1 strongly EFX-orientable if and only if, for every subgraph $H \subseteq G$ such that $H$ is a forest consisting of trees $T_1, T_2...T_k$, for every $1 \leq i \leq k$ there exists $x_i \in T_i$ such that $\bigcup_{i=1}^k N_H(x_i)$ forms an independent set on $G$, where $N_H(x)$ denotes the neighbors of $x$ in $H$ (Figure \ref{fig:2} demonstrates this condition for a $K_{2,4}$ graph). 
\end{lemma}
  \begin{figure}[h]
        \centering
        \includegraphics[width=0.4\textwidth]{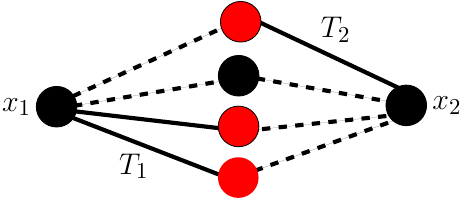}
\caption{   
This figure provides a $K_{2,4}$ graph as an example to demonstrate the characterization of Lemma \ref{lem:01char}. If an adversary chooses $T_1$ and $T_2$ as the forest $H$, we can respond by choosing $x_1 \in T_1$ and $x_2 \in T_2$, and $\bigcup_{i=1}^k N_H(x_i)$ would be the red vertices, which is indeed an independent set on $G$. In general, for any forest an adversary chooses on a $K_{2,4}$ (or any bipartite graph), if a tree is just a single vertex its neighbors in the forest form the empty set, so we can choose the vertex and ignore it. Otherwise, we can choose all $x_i$ to be in the same partite set, and its neighbors in the forest will all be from of the same partite set and hence are an independent set on $G$, so this is an example of a graph that fulfills this condition.
}
        \label{fig:2}
   \end{figure}
   
  \begin{proof}
        (if direction) Suppose an adversary gives us a 0-1 additive valuation assignment a graph which satisfies these conditions. For all the edges with asymmetric valuation (where the two endpoints do not value the edge equally), orient them towards the vertex that values the edge. Call the vertices that receive the asymmetrically valued edges \emph{special} vertices. Let $H'$ be the subgraph consisting of edges valued at 1 for both endpoints. For each component of $H'$ that contains a special vertex, create a spanning tree with the special vertex as the root, and give each vertex the edge from its parent. For a component $C \subseteq H'$ that contains a cycle but has no special vertex, we remove one edge $uv \in C$ where $uv$ is in a cycle, and construct a spanning tree through $C-uv$ with $v$ as the root. Assign each vertex in $C$ the edge from its parent, and assign $uv$ to $v$. Orient the rest of the edges in these components arbitrarily. Since all vertices in such components receive an edge they value, they do not envy anyone else, and are hence not envied by any vertex. \\
        Let $H=T_1,T_2,...T_k$ be the components of $H'$ that are trees which do not have a special vertex. For $H=T_1,T_2,...T_k$, retrieve the vertices $x_1,x_2...x_k$ such that $\bigcup_{i=1}^k N_H(x_i)$ forms an independent set in $G$. For each $T_i$, set $x_i$ to be the root, and give each vertex in $T_i$ the edge from its parent. Since all vertices except $x_1,x_2,...x_k$ in $H$ received at least one edge of value 1, the only vertices envied are $\bigcup_{i=1}^k N_H(x_i)$. We now need to orient the edges of weight 0. Every edge of weight 0 has a non-envied endpoint, as $\bigcup_{i=1}^k N_H(x_i)$ is an independent set on $G$, so no edges can exist between two envied vertices, hence we can orient the remaining edges towards a non-envied endpoint. Since the only envied vertices have exactly one item, this orientation is EFX. 
        \medskip
        
        (only if direction) Suppose there exists a forest $H$ such that for its trees $T_1, T_2...T_k$, for every collection of $x_1,x_2...x_k$ such that $x_i \in T_i$, their neighborhoods in $H$ do not form an independent set on $G$. Set all the edges in $H$ to have weight 1 for both vertices, and set the rest of the edges to have weight 0 for both vertices. Note that any orientation on a tree must have a source, so for any orientation in $G$, we can choose a collection of sources in $s_1, s_2...s_k$ in $H$ such that $s_i \in T_i$ for $1 \leq i \leq k$. Since each $s_i$ receives nothing of value and has all its valued edges go to its neighbors, all vertices in $N_H(s_i)$ are envied by $s_i$, so all vertices in $\bigcup_{i=1}^k N_H(s_i)$ must be envied. However, since $\bigcup_{i=1}^k N_H(s_i)$ does not form an independent set on $G$, then an edge of weight 0 must be between two vertices in that set. Assigning that edge will break EFX since both its endpoints are envied due to the orientation on $H$, hence such a graph is not 0-1 strongly EFX-orientable. 
    \end{proof}

    We now proceed with a simpler condition that is necessary (but not sufficient) for 0-1 strong EFX-orientability, and come up with a method to generate many graphs that violate this condition. 
\begin{corollary}\label{cor:01nec}
    Let $G$ be a 0-1 strongly EFX orientable graph. For any matching $M$ on $G$, the subgraph induced by the vertices of $M$ has an independent set of size $|M|$. 
\end{corollary}
    \begin{proof}
        Let $M$ consist of edges $e_1, e_2, ....e_k$. Each edge in the matching can be considered a tree, and the edges form a forest. Hence, by Lemma \ref{lem:01char}, for each edge $e_i$ in $M$, we can select $v_i \in e_i$ such that $\bigcup_{i=1}^k N_M(v_i)$ is an independent set. Since each vertex in $M$ has exactly one other neighbor also in $M$, then $\bigcup_{i=1}^k N_M(v_i)$ consists of $|M|$ vertices, and hence is an independent set of size $|M|$. 
    \end{proof}

Our next lemma gives us a method to generate more graphs that violate the conditions of the previous corollary by using subdivisons. The \emph{subdivision} of an edge $uv \in E(G)$ is defined as deleting the edge $uv$, adding a new vertex $w$, and adding the edges $uw$ and $wv$ in $G$, effectively placing a vertex in the middle of an edge and dividing it into two edges. A graph $G'$ is said to be a subdivision of $G$ if $G'$ can be obtained through an iteration of edge subdivisions starting from $G$.

\begin{lemma}\label{lem:subdiv}
    Let $G$ be a graph that violates the matching condition in Corollary \ref{cor:01nec}. Let $G'$ be the graph obtained by subdividing any edge of $G$ twice. Then $G'$ violates the matching condition of Corollary \ref{cor:01nec} as well. 
\end{lemma}
    \begin{proof}
        Let the subdivided edge on $G'$ consist of $u,x_1,x_2,v$, where $uv$ is the original edge in $G$ and $x_1,x_2$ are the added vertices. Suppose that for any matching in $G'$, there exists an independent set of the same size on the graph induced by the vertices of that matching. Take a matching $M$ on $G$. We split this problem into two cases: \\
        Case 1: $uv$ is not in $M$. Consider the matching $M'$ on $G'$, which we construct by taking all the edges in $M$ and the edge $x_1x_2$. We can find an independent set $I'$ of size $|M|+1$ in the subgraph induced by the vertices of $M'$. We claim that $I'-\{ x_1, x_2 \}$ is an independent set on $G'$. Indeed, since all vertices other than $u$ and $v$ have the same neighborhoods, we only need to check $u$ and $v$ to ensure that both weren't selected. Since exactly one of $x_1,x_2$ must be in $M'$ in order for $I'$ to have a size of $|M|+1$, if both $u$ and $v$ were in $I'$, then $I'$ cannot be an independent set. Hence, $M$ has a corresponding independent set of size $|M|$ on the subgraph induced by its vertices on $G$. \\
        Case 2: $uv$ is $M$. Consider the matching $M'$ on $G'$, which we construct by taking $M-uv$ and the edges $ux_1$ and $x_2v$. There exists an independent set $I'$ of size $|M|+1$ on the subgraph induced by $M'$. We claim that at least one of $u,v$ must be in $I'$. Indeed, if neither $u$ nor $v$ is in $I'$, then both $x_1$ and $x_2$ must be in $I'$, which is impossible since they share an edge. Without loss of generality, suppose that $v \in I'$. We claim that $I'-\{ u,x_1\}$ is an independent set on $G$. This is true since all vertices in $I'-\{ u,x_1\}$ other than $v$ have the same neighborhood, and $u$ cannot be in $I'-\{ u,x_1\}$, meaning that $I'-\{ u,x_1\}$ is an independent set of of size $|M|$, as desired. 
    \end{proof}

Finally, we prove that the chromatic number of a 0-1 strongly EFX-orientable graph, and hence a strongly EFX-orientable graph, is upper bounded by 3. To do this, we will need to use the following result from Zang \cite{Zang1998} and Thomassen \cite{Thomassen2001}.
\begin{lemma}[\cite{Zang1998,Thomassen2001}]\label{lem:ZT}
    A graph $G$ of $\chi(G) \geq 4$ contains a subdivision of a $K_4$ where each edge of the $K_4$ corresponds to a path of odd length (also known as a totally odd subdivision). 
\end{lemma}

\begin{lemma}\label{lem:01UB}
    If a graph $G$ is 0-1 strongly EFX-orientable, then $\chi(G) \leq 3$. 
\end{lemma}    
    \begin{proof}
        %
        We claim that $K_4$ violates the matching condition in Corollary \ref{cor:01nec}. 
        Indeed, there is no independent set of size 2 on a $K_4$, while it's possible to have a matching of size 2 in a $K_4$. By Lemma \ref{lem:subdiv}, a totally odd subdivision of a graph that violates the matching condition will violate the matching condition itself as well. Hence, 0-1 strongly EFX-orientable graphs cannot contain a totally odd $K_4$ subdivision, so $\chi(G) \leq 3$ by Lemma \ref{lem:ZT}. 
    \end{proof}

The next theorem follows as a corollary of Lemma \ref{lem:01UB}.
\begin{theorem}
    \label{thm:UB}
    If a graph is strongly EFX-orientable, then $\chi(G) \leq 3$.
   
\end{theorem}
    \begin{proof}
    Any strongly EFX-orientable is 0-1 strongly EFX-orientable, so $\chi(G) \leq 3$ follows from Lemma \ref{lem:01UB}. 
    \end{proof}

\section{Strong EFX-Orientability: Sufficient Condition via Bipartiteness}\label{sec:suf}
In this section, we show that if the graph is almost bipartite, then it is strongly EFX-orientable. This in turn implies that graphs with chromatic number at most two are strongly EFX orientable.

\begin{lemma}\label{lem:suf1}
    Any bipartite graph is strongly EFX-orientable. 
\end{lemma}
    \begin{proof}
        Given a bipartite graph $G$, partition the vertices into two color classes $A$ and $B$. Let every vertex in $A$ pick their favorite edge, and orient the chosen edges towards $A$. Note that no two vertices in $A$ pick the same edge, since $A$ is an independent set. Orient all remaining edges towards $B$. No envy exists within an independent set, since any two vertices in an independent set value no common edges. For any vertex $a \in A$, every vertex in $B$ receives at most one edge adjacent to $a$, and since $a$ was allocated its favorite edge, $a$ does not envy any vertex in $B$. Note that while there might be envy from $B$ to $A$, since every vertex in $A$ has only edge, the envy disappears after removal of this edge. Hence, this allocation is EFX. 
    \end{proof}

It is possible to relax the bipartite condition slightly to obtain a slightly weaker sufficient condition for strong EFX-orientability. This relaxation gives us examples of graphs with chromatic number 3 that are also strongly EFX-orientable. 

\begin{lemma}\label{lem:suf2}
    If $G$ is a graph with some $v \in G$ such that for any edge $e$ incident to $v$, $G-e$ is bipartite, then $G$ is strongly EFX-orientable. 
\end{lemma}
    \begin{proof}
        Let such a vertex $v$ pick its favorite edge, which we will call $uv$. Delete $uv$ from $G$. Note that the remaining graph is bipartite. Partition $G-uv$ into color classes $A$ and $B$, and without loss of generality let $v$ be in $A$. Have all vertices in $A-v$ pick their favorite edge, and orient the chosen edges towards $A$. Orient the remaining edges towards $B$.
        
        Note that $B$ must be an independent set in $G$, since $uv$ has one end in $A$, and all other edges go from $A$ to $B$, so no envy exists between vertices in $B$. No vertex in $A$ envies a vertex in $B$, since they got to take an edge before $B$, and the remaining edges incident to a vertex in $A$ were assigned to different neighbors of that vertex. While there might be envy towards vertices in $A$, every vertex in $A$ was assigned only one edge, making this orientation EFX. 
    \end{proof}

\begin{corollary}\label{cor:suf}
    Odd cycles with an additional edge are strongly EFX-orientable. 
\end{corollary}
    \begin{proof}
        Adding an additional edge to an odd cycle splits the graph into a smaller even cycle and a smaller odd cycle joined by an edge. We can take a vertex $v$ in the smaller odd cycle but not in the smaller even cycle. Deleting any edge incident to $v$ removes the graph of all its odd cycles, hence by Lemma \ref{lem:suf2}, an odd cycle with a chord is strongly EFX-orientable. 
    \end{proof}

We conclude this subsection with a remark that when dealing with strongly EFX-orientable graphs, if needed, we can always assume that a graph has minimum degree $\delta(G) \geq 2$, as vertices with only one incident edge do not affect strong EFX-orientability.  
\begin{proposition}\label{prop:suff}
    Let $G$ be a graph with a vertex $v$ that has degree 1. $G-v$ is strongly EFX-orientable if and only if $G$ is strongly EFX orientable. 
\end{proposition}
    \begin{proof}
        $G$ being strongly EFX orientable implies that $G-v$ is strongly EFX orientable, so we prove the other direction. Given any valuation assignment on $G$, orient the edge $e$ incident to $v$ towards $v$, and orient $G-v$ such that the orientation is EFX on $G-v$. Note that while the neighbor of $v$ could envy $v$, the orientation is still EFX, as $v$ only has one item.  
    \end{proof}

\section{General Valuations: Forbidden Structures}\label{sec:fs}

To prove that our upper bound of $\chi(G) \leq 3$ is sharp, we give various examples of graphs with chromatic number 3 which are not strongly EFX-orientable. We also show that the bipartite condition for sufficiency can be quite brittle, as adding an edge between two vertices of the same partite set for many bipartite graphs may cause the graph to no longer be strongly EFX-orientable. 

\begin{proposition}\label{prop:fs1}
    Two triangles glued together at a vertex, or connected by a path of length 1 or 2 are not strongly EFX-orientable. That is, graphs depicted in Figures \ref{fig:test} and \ref{fig:1} are not strongly EFX-orientable.

\begin{figure}[h]
\centering
\begin{subfigure}{.5\textwidth}
  \centering
  \includegraphics[width=.55\textwidth]{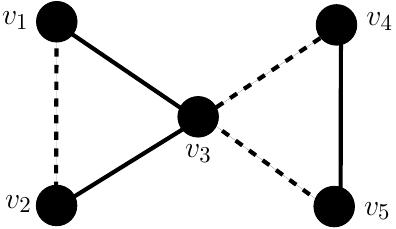}
  \label{fig:sub1}
\end{subfigure}%
\begin{subfigure}{.5\textwidth}
  \centering
  \includegraphics[width=.7\textwidth]{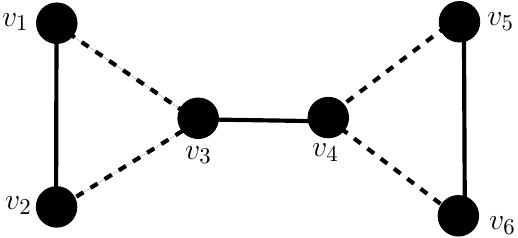}
  \label{fig:sub2}
\end{subfigure}
\caption{}
\label{fig:test}
\end{figure}

    \begin{figure}[h]
        \centering
        \includegraphics[width=0.5\textwidth]{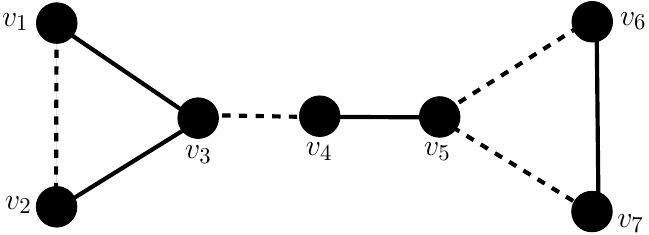}
        \caption{}
        \label{fig:1}
    \end{figure}
\end{proposition}
    \begin{proof}
    All three of these examples are not 0-1 strongly EFX orientable and hence not strongly EFX orientable. The figures above give an example of which edges to assign a value of 1 (solid lines) and which edges to assign a value of 0 to (dashed lines) for a graph that is not strongly EFX-orientable. Note that in the case of two triangles joined by an edge, a matching of size 3 exists, but the maximum independent set is at most 2, since if 3 vertices are selected two of them will be from the same triangle. Hence, by Corollary \ref{cor:01nec}, this graph is not 0-1 strongly EFX orientable. 

    Consider the forest marked by the solid lines on the graph of two triangles glued by a vertex. It consists of two trees, $T_1=v_1v_3v_2$ and $T_2=v_4v_5$. We claim that it is impossible to select an $x_1 \in T_1$ and $x_2 \in T_2$ such that their neighbors form an independent set on the entire graph. Indeed, if this would be possible, $v_3$ could not be selected as $x_1$, since $v_1$ and $v_2$ are adjacent. Hence, either $v_1$ or $v_2$ needs to be selected, so $v_3$ will always be a neighbor of $x_1$. However, this makes selecting $x_2$ impossible, since $v_3$ is adjacent to both $v_4$ and $v_5$, hence this graph is not 0-1 strongly EFX orientable. 

    Finally, the graph of two triangles connected by a path of length 2 has a forest selected that consists of three trees, $T_1=v_1v_3v_2$, $T_2=v_4v_5$, and $T_3=v_6v_7$. Similar to the previous case, to select $x_1,x_2,x_3$ such that their neighborhoods in the forest form an independent set, $x_1$ cannot equal $v_3$ and must be either $v_1$ or $v_2$. This means that $x_2$ cannot equal $v_4$, since $v_3$ and $v_4$ are adjacent, so $x_2=v_5$. But both $v_6$ and $v_7$ are adjacent to $v_5$, so selecting $x_3$ is impossible, meaning that this graph is not 0-1 strongly EFX orientable as well. 
    \end{proof}

To separate the notions of 0-1 strong EFX orientability and strong EFX orientability in general, we provide an example of a graph that is 0-1 strongly EFX-orientable but not strongly EFX-orientable. 
\begin{proposition}\label{prop:fs2}
    Let $G$ consist of two triangles sharing one edge, as shown in Figure \ref{fig:3}. Then $G$ is not strongly EFX-orientable. However, $G$ is 0-1 strongly EFX orientable. 
    \begin{figure}[h]
        \centering
        \includegraphics[width=0.25\textwidth]{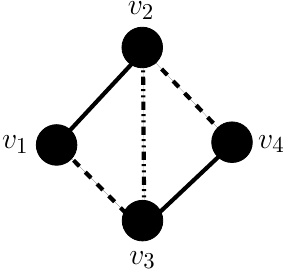}
        \caption{}
        \label{fig:3}
    \end{figure}
\end{proposition}
    \begin{proof}
        Let the first triangle be $v_1v_2v_3$, and the second triangle be $v_2v_3v_4$. Let $v_2v_3$ have weight 0.5 for both endpoints, $v_1v_2$ and $v_3v_4$ have weight 1 for both endpoints, and the rest of the edges have weight 0 for both endpoints. All valuations are additive. Due to the symmetry of the graph, without loss of generality, we can assume that $v_2v_3$ is directed towards $v_2$. For the orientation to be EFX, $v_2$ cannot receive $v_1v_2$, as $v_1$ will envy $v_2$ even though $v_2$ has already received $v_2v_3$. Hence $v_1$ must receive $v_1v_2$, but $v_1$ will be envied by $v_2$, so $v_1v_3$ must be directed towards $v_3$. Since $v_3$ already has an edge, for similar reasons, $v_3$ cannot receive $v_3v_4$, so $v_3v_4$ is directed towards $v_4$. Note that $v_2$ and $v_4$ are both envied by $v_3$, but there is another edge between $v_2$ and $v_4$ that needs to be assigned, making an EFX-orientation impossible. \\
        We now prove that this graph is 0-1 strongly EFX-orientable. We wish to prove that for any forest $H=T_1,T_2...T_k$ on this graph, we can find a collection of vertices $x_1 \in T_1, x_2 \in T_2... x_k \in T_k$ such that $\bigcup_{i=1}^k N_H(x_i)$ is an independent set on $G$. We can assume, without loss of generality, that all trees in $H$ have at least one edge, because if a tree $T_i$ is a single vertex, its neighborhood in $H$ is the empty set, and selecting it as $x_i$ contributes nothing to $\bigcup_{i=1}^k N_H(x_i)$. Observe that the only way to select two trees which both have at least one edge on this graph is a perfect matching, which we can select an independent set of size 2 from. Any other forest where all the trees have edges will have at most one tree, and we can select a leaf from the tree as $x_1$ - its single neighbor in the forest will trivially form an independent set on $G$. Hence this graph satisfies the condition in Lemma \ref{lem:01char} and is 0-1 strongly EFX orientable. 
    \end{proof}

In order to generate more forbidden subgraphs, we need a lemma similar to Lemma \ref{lem:subdiv} that allows us to subdivide certain edges of graphs that are not strongly EFX-orientable to create a new graph that is not strongly EFX-orientable. To help us do this, we introduce the idea of an item having \emph{zero value} to an agent. Formally, if $M$ is the set of a goods, an item $m \in M$ has \emph{zero value} to agent $i$ if for all $X \subseteq M$, $f_i(X-\{m\}) = f_i(X)$. 

\begin{lemma}\label{lem:fs1}
    Let $G$ be a graph, and suppose there exists a valuation on $G$ that is not EFX-orientable, and there exists an edge $e$ in the valuation that has zero value for both endpoints. Construct $G'$ by replacing $e$ with any path of odd length. Then $G'$ is not strongly EFX-orientable.  
\end{lemma}
    \begin{proof}
        Let the endpoints of $e$ be $u$ and $v$. We first consider the case where we subdivide $e$ twice with vertices $x_1$ and $x_2$ to form $G'$. $ux_1$ and $vx_2$ will have zero weight for both endpoints, while $x_1x_2$ will have a weight of 1 both endpoints. We claim that $G'$ does not have an EFX orientation with this valuation, so suppose, for contradiction, that $G'$ actually does have such an orientation. 

        Note that if $ux_1$ and $vx_2$ are oriented towards $x_1$ and $x_2$ respectively, then this orientation cannot be EFX, as whichever vertex that receives $x_1x_2$ will be envied by the other vertex, so they can't receive any additional edges for the orientation to be EFX. Hence, at least one of $u$ and $v$ need to receive an edge created by subdividing $e$. Without loss of generality, suppose that in the EFX orientation, $v$ receives $vx_2$. 

        Returning to $G$, we can copy the EFX orientation on $G'$, directing $e$ towards $v$. Note that all vertices in $G-u$ have the same bundle as they had on $G'$, meaning that the allocation within $G-u$ is EFX. While it's possible that $u$ could've received $ux_1$, since $ux_1$ had zero value, removing it does not change the value of the bundle $u$ receives, hence $u$ values the bundles it receives on $G$ and $G'$ equally. Furthermore, since $uv$ has zero value, then $u$ will not value the bundle given to $v$, hence for any vertex on $G-u$, $u$ will not envy any strict subset of that vertex's bundle. Since the bundle $u$ receives at $G$ is a subset of what it received on $G'$, no vertex from $G-u$ envies a strict subset of the bundle $u$ receives either. Hence there exists an EFX orientation on $G$ with such a valuation function, which is a contradiction. 

        Note while subdividing $e$, we created edges $ux_1$ and $vx_2$ which also have zero value, meaning that we can repeat this process again by subdividing any one of those edges twice, hence replacing $e$ with an odd path of edges alternating between zero and one values. The new graph will not have an EFX orientation with this valuation, and hence is not strongly EFX-orientable, as desired. 
    \end{proof}

\begin{corollary}\label{cor:fs1}
    Odd cycles that share exactly one edge, odd cycles that share exactly one vertex, and disjoint odd cycles connected by a path are not strongly EFX-orientable.
\end{corollary}    
    \begin{proof}
        We first handle the case of odd cycles joined by a path of odd length. By Proposition \ref{prop:fs1}, two triangles connected by a path of length 1 violated the matching condition in Corollary \ref{cor:01nec}, and hence by Lemma \ref{lem:subdiv}, we can replace the any edge with an odd path of arbitrary length by repeated subdivisions, allowing us to generate odd cycles of any size connected by a path of any odd length, all of which are not strongly EFX orientable. 
        
        For the case of of odd cycles joined by a path of even length, or glued by one edge or vertex, note that by Proposition \ref{prop:fs1} and \ref{prop:fs2}, triangles joined by a path of length 2, and triangles glued by one edge or vertex are not strongly EFX orientable. Furthermore, the examples of bad valuation assignments in all cases had at least one edge in the odd cycle have zero value, meaning that we can replace the triangle with an odd cycle of any length and still have a non strongly EFX orientable graph by Lemma \ref{lem:fs1}. Similarly, in the case of two triangles being joined by a path of length 2, one of the edges in the path had zero value, allowing us to replace that edge with a path of odd length and giving us an even path of any length, as desired. 
    \end{proof}

Finally, the next lemma shows that the bipartite condition cannot be relaxed much, as adding a single edge to many bipartite graphs can break strong EFX-orientability. 

\begin{lemma}\label{lem:suf-tight}
    Let $G$ be a bipartite graph with $|V(G)| \geq 4$ that remains connected after the deletion of any vertex. Suppose that an edge $e$ is added, which connects two vertices of the same partite set. Then $G+e$ is not strongly EFX-orientable.  
\end{lemma}    
\begin{proof}
        Let $u$ and $v$ be the vertices incident to $e$. By Menger's Theorem, since $G$ has no cut vertices, there exists a cycle on $G$ through $u$ and $v$. Since $G$ is bipartite, such a cycle must have an even number of vertices, and since $u$ and $v$ are part of the same partite set, then their distance on the cycle must be even. Adding $e$ onto the cycle breaks it into two odd cycles which share $e$ as an edge, which is not strongly EFX-orientable, hence $G$ is not strongly EFX-orientable. 
    \end{proof}

\section{Discussion}\label{sec:disc}
In this paper, we studied {\em strong EFX-orientability} for graphical valuations \cite{Christodoulou2023}. We showed a deep connection of this property to the chromatic number of the graph: every strongly EFX-orientable graph has a chromatic number of at most three, and a graph with a chromatic number of two or less is strongly EFX-orientable. This result is tight in the following sense where we demonstrate a 3-chromatic graph with and without this property. The upper bound of 3 on the chromatic number demonstrates that apart from a niche class of graphs with nice structure, most graphs are not strongly EFX orientable. Therefore, it is reasonable to expect that when forming an EFX allocation in this graphical setting, not all edges are allocated to an incident vertex that values them. 


We also demonstrated that 0-1 strong EFX orientability and general strong EFX orientability were separate notions, and it remains open as to whether there is a graph that is EFX orientable for all additive valuations but not EFX orientable for a monotone function, or if additive strong EFX-orientability and strong EFX-orientability are equivalent. When given a graph with the valuation functions, determining whether an EFX-orientation exists is NP-hard \cite{Christodoulou2023}, but it is currently unknown if determining whether a graph is strongly EFX-orientable or not is NP-hard. It would be interesting to characterize graphs which admit an EFX orientation that is also Pareto optimal (PO). Finally, we hope that our results will help with obtaining a complete characterization of strongly EFX-orientable graphs. 

\section*{Acknowledgements}
The first author would like to thank Elfarouk Harb and Daniel Schoepflin for a careful reading of this manuscript, and Pooja Kulkarni for thoughtful discussions.


\newpage

\printbibliography

@Article{Zang1998,
author={Zang, Wenan},
title={Proof of Toft's Conjecture: Every Graph Containing No Fully Odd K4 is 3-Colorable},
journal={Journal of Combinatorial Optimization},
year={1998},
month={Jun},
day={01},
volume={2},
number={2},
pages={117-188}
}

@PhdThesis{Vossen02,
  author = {T. W.M. Vossen}, 
  title = {Fair allocation concepts in air traffic management}, 
  school = {University of Maryland, College Park},
  year = {2002}
}

@article{BudishC10,
author = {Budish, Eric and Cantillon, Estelle},
year = {2010},
title = {The Multi-Unit Assignment Problem: Theory and Evidence from Course Allocation at {H}arvard},
volume = {102},
journal = {American Economic Review} 
}

@InProceedings{EtkinPT05,
 	author = {R. Etkin and A. Parekh and D. Tse}, 
	title = {Spectrum Sharing for Unlicensed Bands}, 
	booktitle = {In Proceedings of the first IEEE Symposium on New Frontiers in Dynamic Spectrum Access Networks}, 
	year = {2005}
}

@book{BramsT96,
  author    = {Steven J. Brams and
               Alan D. Taylor},
  title     = {Fair division - from cake-cutting to dispute resolution},
  publisher = {Cambridge University Press},
  year      = {1996}
}

@article{PrattZ90,
 author = {Pratt, John Winsor and Zeckhauser, Richard Jay},
 title = {The Fair and Efficient Division of the {W}insor Family Silver},
 journal = {Management Science},
 volume = {36},
 number = {11},
 year = {1990},
 pages = {1293--1301}
}

@article{Steinhaus48,
	author={Hugo Steinhaus},
	journal = {Econometrica},
	number = {1},
	pages = {101-104},
	title = {The Problem of Fair Division},
	volume = {16},
	year = {1948}
}

@article{ProcaciaCACM,
	author = {Procaccia, Ariel D.},
	title = {Technical Perspective: An Answer to Fair Division's Most Enigmatic Question},
	year = {2020},
	issue_date = {April 2020},
	publisher = {Association for Computing Machinery},
	address = {New York, NY, USA},
	volume = {63},
	number = {4},
	journal = {Commun. ACM},
	month = mar,
	pages = {118},
	numpages = {1}
}

@article{PR18,
	author    = {Benjamin Plaut and
	Tim Roughgarden},
	title     = {Almost Envy-Freeness with General Valuations},
	journal   = {{SIAM} J. Discret. Math.},
	volume    = {34},
	number    = {2},
	pages     = {1039--1068},
	year      = {2020}
}

@Article{Thomassen2001,
author={Thomassen, Carsten},
title={Totally Odd -subdivisions in 4-chromatic Graphs},
journal={Combinatorica},
year={2001},
month={Jul},
day={01},
volume={21},
number={3},
pages={417-443}
}

@inproceedings{Christodoulou2023,
author = {Christodoulou, George and Fiat, Amos and Koutsoupias, Elias and Sgouritsa, Alkmini},
title = {Fair allocation in graphs},
year = {2023},
publisher = {Association for Computing Machinery},
address = {New York, NY, USA},
booktitle = {Proceedings of the 24th ACM Conference on Economics and Computation},
pages = {473–488},
numpages = {16},
location = {, London, United Kingdom, },
series = {EC '23}
}

@inproceedings{Chaudhury2020,
author = {Chaudhury, Bhaskar Ray and Garg, Jugal and Mehlhorn, Kurt},
title = {EFX Exists for Three Agents},
year = {2020},
publisher = {Association for Computing Machinery},
address = {New York, NY, USA},
booktitle = {Proceedings of the 21st ACM Conference on Economics and Computation},
pages = {1–19},
numpages = {19},
location = {Virtual Event, Hungary},
series = {EC '20}
}

@inproceedings{Akrami23,
author = {Akrami, Hannaneh and Alon, Noga and Chaudhury, Bhaskar Ray and Garg, Jugal and Mehlhorn, Kurt and Mehta, Ruta},
title = {EFX: A Simpler Approach and an (Almost) Optimal Guarantee via Rainbow Cycle Number},
year = {2023},
publisher = {Association for Computing Machinery},
address = {New York, NY, USA},
booktitle = {Proceedings of the 24th ACM Conference on Economics and Computation},
pages = {61},
numpages = {1},
location = {, London, United Kingdom, },
series = {EC '23}
}

@String{AAAI     = {Conf.\ Artif.\ Intell.\ (AAAI)}}

@String{AAMAS    = {Conf.\ Auton.\ Agents and Multi-Agent Systems (AAMAS)}}

@String{EC       = {Conf.\ Economics and Computation (EC)}}

@String{ESA      = {European Symp.\ Algorithms (ESA)}}

@String{IJCAI    = {Intl.\ Joint Conf.\ Artif.\ Intell.\ (IJCAI)}}

@String{SODA     = {Symp.\ Discrete Algorithms (SODA)}}

@article{Caragiannis19,
  author       = {Ioannis Caragiannis and
                  David Kurokawa and
                  Herv{\'{e}} Moulin and
                  Ariel D. Procaccia and
                  Nisarg Shah and
                  Junxing Wang},
  title        = {The Unreasonable Fairness of Maximum Nash Welfare},
  journal      = {{ACM} Trans. Economics and Comput.},
  volume       = {7},
  number       = {3},
  pages        = {12:1--12:32},
  year         = {2019}
}

@misc{feldman2023,
      title={On Optimal Tradeoffs between EFX and Nash Welfare}, 
      author={Michal Feldman and Simon Mauras and Tomasz Ponitka},
      year={2023},
      eprint={2302.09633},
      archivePrefix={arXiv},
      primaryClass={cs.GT}
}

@article{akrami2022efx,
  title={EFX allocations: Simplifications and improvements},
  author={Akrami, Hannaneh and Alon, Noga and Chaudhury, Bhaskar Ray and Garg, Jugal and Mehlhorn, Kurt and Mehta, Ruta},
  journal={arXiv preprint arXiv:2205.07638},
  year={2022}
}

@inproceedings{berger2022almost,
  title={Almost full EFX exists for four agents},
  author={Berger, Ben and Cohen, Avi and Feldman, Michal and Fiat, Amos},
  booktitle={Proceedings of the AAAI Conference on Artificial Intelligence},
  volume={36},
  number={5},
  pages={4826--4833},
  year={2022}
}

@article{CKMS21,
	author    = {Bhaskar Ray Chaudhury and
	Telikepalli Kavitha and
	Kurt Mehlhorn and
	Alkmini Sgouritsa},
	title     = {A Little Charity Guarantees Almost Envy-Freeness},
	journal   = {{SIAM} J. Comput.},
	volume    = {50},
	number    = {4},
	pages     = {1336--1358},
	year      = {2021}
}

@inproceedings{CaragiannisGH19,
  author       = {Ioannis Caragiannis and
                  Nick Gravin and
                  Xin Huang},
  editor       = {Anna Karlin and
                  Nicole Immorlica and
                  Ramesh Johari},
  title        = {Envy-Freeness Up to Any Item with High Nash Welfare: The Virtue of
                  Donating Items},
  booktitle    = {Proceedings of the 2019 {ACM} Conference on Economics and Computation,
                  {EC} 2019, Phoenix, AZ, USA, June 24-28, 2019},
  pages        = {527--545},
  publisher    = {{ACM}},
  year         = {2019}
}

@article{GhosalPNV23,
  author       = {Pratik Ghosal and
                  Vishwa Prakash HV and
                  Prajakta Nimbhorkar and
                  Nithin Varma},
  title        = {{EFX} Exists for Four Agents with Three Types of Valuations},
  journal      = {CoRR},
  volume       = {abs/2301.10632},
  year         = {2023},
  url          = {https://doi.org/10.48550/arXiv.2301.10632},
  doi          = {10.48550/arXiv.2301.10632},
  eprinttype    = {arXiv},
  eprint       = {2301.10632},
  bibsource    = {dblp computer science bibliography, https://dblp.org}
}

@inproceedings{Mahara21,
  author       = {Ryoga Mahara},
  editor       = {Petra Mutzel and
                  Rasmus Pagh and
                  Grzegorz Herman},
  title        = {Extension of Additive Valuations to General Valuations on the Existence
                  of {EFX}},
  booktitle    = {29th Annual European Symposium on Algorithms, {ESA} 2021, September
                  6-8, 2021, Lisbon, Portugal (Virtual Conference)},
  series       = {LIPIcs},
  volume       = {204},
  pages        = {66:1--66:15},
  publisher    = {Schloss Dagstuhl - Leibniz-Zentrum f{\"{u}}r Informatik},
  year         = {2021},
  url          = {https://doi.org/10.4230/LIPIcs.ESA.2021.66},
  doi          = {10.4230/LIPIcs.ESA.2021.66},
  bibsource    = {dblp computer science bibliography, https://dblp.org}
}

@inproceedings{GorantlaMV23,
  author       = {Pranay Gorantla and
                  Kunal Marwaha and
                  Santhoshini Velusamy},
  editor       = {Nikhil Bansal and
                  Viswanath Nagarajan},
  title        = {Fair allocation of a multiset of indivisible items},
  booktitle    = {Proceedings of the 2023 {ACM-SIAM} Symposium on Discrete Algorithms,
                  {SODA} 2023, Florence, Italy, January 22-25, 2023},
  pages        = {304--331},
  publisher    = {{SIAM}},
  year         = {2023},
  url          = {https://doi.org/10.1137/1.9781611977554.ch13},
  doi          = {10.1137/1.9781611977554.ch13},
  bibsource    = {dblp computer science bibliography, https://dblp.org}
}

@inproceedings{LivanosMM22,
  author       = {Vasilis Livanos and
                  Ruta Mehta and
                  Aniket Murhekar},
  editor       = {Piotr Faliszewski and
                  Viviana Mascardi and
                  Catherine Pelachaud and
                  Matthew E. Taylor},
  title        = {(Almost) Envy-Free, Proportional and Efficient Allocations of an Indivisible
                  Mixed Manna},
  booktitle    = {21st International Conference on Autonomous Agents and Multiagent
                  Systems, {AAMAS} 2022, Auckland, New Zealand, May 9-13, 2022},
  pages        = {1678--1680},
  publisher    = {International Foundation for Autonomous Agents and Multiagent Systems
                  {(IFAAMAS)}},
  year         = {2022},
  url          = {https://www.ifaamas.org/Proceedings/aamas2022/pdfs/p1678.pdf},
  doi          = {10.5555/3535850.3536074},
  bibsource    = {dblp computer science bibliography, https://dblp.org}
}

@article{AmanatidisMN20,
	author    = {Georgios Amanatidis and
	Evangelos Markakis and
	Apostolos Ntokos},
	title     = {Multiple birds with one stone: Beating 1/2 for {EFX} and {GMMS} via
	envy cycle elimination},
	journal   = {Theor. Comput. Sci.},
	volume    = {841},
	pages     = {94--109},
	year      = {2020}
}

@article{Mahara23,
  author       = {Ryoga Mahara},
  title        = {Existence of {EFX} for two additive valuations},
  journal      = {Discret. Appl. Math.},
  volume       = {340},
  pages        = {115--122},
  year         = {2023},
  url          = {https://doi.org/10.1016/j.dam.2023.06.035},
  doi          = {10.1016/j.dam.2023.06.035},
  bibsource    = {dblp computer science bibliography, https://dblp.org}
}

@inproceedings{BuSongYu23,
  author       = {Xiaolin Bu and
                  Jiaxin Song and
                  Ziqi Yu},
  editor       = {Minming Li and
                  Xiaoming Sun and
                  Xiaowei Wu},
  title        = {{EFX} Allocations Exist for Binary Valuations},
  booktitle    = {Frontiers of Algorithmics - 17th International Joint Conference, {IJTCS-FAW}
                  2023 Macau, China, August 14-18, 2023, Proceedings},
  series       = {Lecture Notes in Computer Science},
  volume       = {13933},
  pages        = {252--262},
  publisher    = {Springer},
  year         = {2023},
  url          = {https://doi.org/10.1007/978-3-031-39344-0\_19},
  doi          = {10.1007/978-3-031-39344-0\_19},
  bibsource    = {dblp computer science bibliography, https://dblp.org}
}

@inproceedings{LiptonMMS04,
  author       = {Richard J. Lipton and
                  Evangelos Markakis and
                  Elchanan Mossel and
                  Amin Saberi},
  editor       = {Jack S. Breese and
                  Joan Feigenbaum and
                  Margo I. Seltzer},
  title        = {On approximately fair allocations of indivisible goods},
  booktitle    = {Proceedings 5th {ACM} Conference on Electronic Commerce (EC-2004),
                  New York, NY, USA, May 17-20, 2004},
  pages        = {125--131},
  publisher    = {{ACM}},
  year         = {2004},
  url          = {https://doi.org/10.1145/988772.988792},
  doi          = {10.1145/988772.988792},
  bibsource    = {dblp computer science bibliography, https://dblp.org}
}

@inproceedings{BKV18,
	author    = {Siddharth Barman and
	Sanath Kumar Krishnamurthy and
	Rohit Vaish}, 
	title     = {Finding Fair and Efficient Allocations},
	booktitle = {Proceedings of the 19th ACM Conference on Economics and Computation (EC)},
	pages     = {557--574},
	year      = {2018}
}

@inproceedings{chaudhury2021improving,
  title={Improving EFX guarantees through rainbow cycle number},
  author={Chaudhury, Bhaskar Ray and Garg, Jugal and Mehlhorn, Kurt and Mehta, Ruta and Misra, Pranabendu},
  booktitle={Proceedings of the 22nd ACM Conference on Economics and Computation},
  pages={310--311},
  year={2021}
}

@inproceedings{chan2019maximin,
author = {Chan, Hau and Chen, Jing and Li, Bo and Wu, Xiaowei},
title = {Maximin-aware allocations of indivisible goods},
year = {2019},
publisher = {AAAI Press},
booktitle = {Proceedings of the 28th International Joint Conference on Artificial Intelligence},
pages = {137–143},
numpages = {7},
location = {Macao, China},
series = {IJCAI'19}
}

\end{document}